\documentclass[runningheads]{llncs}
\usepackage{isabelle,isabellesym}
\usepackage{amsmath, amssymb}
\usepackage{multirow}
\usepackage{graphicx}
\usepackage{url}
\usepackage{breakurl}
\usepackage[breaklinks]{hyperref}
\usepackage{enumerate}
\usepackage{subcaption}
\isabellestyle{it}

\newcommand{\emp}{\varepsilon}

\newcommand{\abs}[1]{\left\vert #1 \right\vert}
\newcommand{\gen}[1]{\left\langle #1 \right\rangle}

\DeclareMathOperator{\lists}{\tt lists}
\DeclareMathOperator{\concat}{\tt concat}
\DeclareMathOperator{\hd}{\tt hd}
\DeclareMathOperator{\last}{\tt last}
\DeclareMathOperator{\R}{\mathcal R}

\newcommand{\ws}{\mathbf w}
\newcommand{\ps}{\mathbf p}

\newcommand{\us}{\mathbf u}
\newcommand{\vs}{\mathbf v}
\newcommand{\zs}{\mathbf z}

\newcommand{\isakw}[1]{\isakeyword{#1}}
\usepackage{todonotes}

\usepackage[most]{tcolorbox}
\newtcolorbox{isaframe}[1][]
   { blanker, %breakable, 
    left=3mm, right=3mm, top=1mm, bottom=1mm,
     borderline west={1pt}{0pt}{blue},
     before upper=\setlength{\parindent}{0pt},
     parbox=true, #1}

\title{Binary codes that do not preserve primitivity}

\author{Štěpán Holub\inst{1} \orcidID{0000-0002-6169-5139}
Martin Raška\inst{1} \orcidID{0000-0002-4414-4538} \and
Štěpán Starosta\inst{2}\orcidID{0000-0001-5962-4297}
}
\authorrunning{M. Raška, Š. Starosta}
\institute{Faculty of Mathematics and Physics, Charles University, Czech Republic\\
\and
Faculty of Information Technology, Czech Technical University in Prague, Czech Republic\\
}

\authorrunning{\v S. Holub and M. Ra\v ska and \v S. Starosta}%TODO mandatory. First: Use abbreviated first/middle names. Second (only in severe cases): Use first author plus 'et al.'

\begin{document}

\maketitle

\begin{abstract}
A code $X$ is not primitivity preserving if there is a primitive list $\ws \in \lists X$ whose concatenation is imprimitive. We formalize a full characterization of such codes in the binary case in the proof assistant Isabelle/HOL. Part of the formalization, interesting on its own, is a description of $\{x,y\}$-interpretations of the square $xx$ if $\abs y \leq \abs x$. We also provide a formalized parametric solution of the related equation $x^jy^k = z^\ell$.
\end{abstract}

% \maketitle
\section{Introduction}
Consider two words ${\tt abba}$ and ${\tt b}$. It is possible to concatenate (several copies of) them as ${\tt b}\cdot {\tt abba} \cdot {\tt b}$, and obtain a power of a third word, namely a square ${\tt bab}\cdot {\tt bab}$ of ${\tt bab}$.
In this paper, we completely describe all ways how this can happen for two words, and formalize it in Isabelle/HOL. 

The corresponding theory has a long history. The question can be formulated as solving equations in three variables of the special form $W(x,y) = z^\ell$ where the left hand side is a sequence of $x$'s and $y$'s, and $\ell \geq 2$. 
The seminal result in this direction is the paper by R. C. Lyndon and M.-P. Sch\"utzen\-ber\-ger \cite{lyndon1962} from 1962, which solves in a more general setting of free groups the equation $x^jy^k = z^\ell$ with $2 \leq j,k,\ell$. It was followed, in 1967, by a partial answer to our question by A. Lentin and M.-P. Sch\"{u}tzenberger \cite{lentin}.
 Complete characterization of monoids generated by three words was provided by L. G. Budkina and Al. A. Markov in 1973 \cite{budkina}. The characterization was later, in 1976, reproved in a different way by Lentin's student J.-P. Spehner in his Ph.D. thesis \cite{spehner}, which even explicitly mentions the answer to the present question. See also a comparison of the two classifications by T. Harju and D. Nowotka \cite{terodirk}. In 1985, the result was again reproved by E. Barbin-Le Rest and M. Le Rest \cite{lerest}, this time specifically focusing on our question. Their paper contains a characterization of binary interpretations of a square as a crucial tool. The latter combinatorial result is interesting on its own, but is very little known. In addition to the fact that, as far as we know, the proof is not available in English, it has to be reconstructed from Th\'eor\`eme 2.1 and Lemme 3.1 in \cite{lerest}, it is long, technical and little structured, with many intuitive steps that have to be clarified. It is symptomatic, for example, that Ma\v nuch \cite{Manuch} cites the claim as essentially equivalent to his desired result but nevertheless provides a different, shorter but similarly technical proof.

This complicated history makes the topic a perfect candidate for formalization. The proof we present here naturally contains some ideas of the proof from \cite{lerest} but is significantly different. Our main objective was to follow the basic methodological requirement of a good formalization, namely to identify claims that are needed in the proof and formulate them as separate lemmas and as generally as possible so that they can be reused not only in the proof but also later. 
Moreover, the formalization naturally forced us to consider carefully the overall strategy of the proof (which is rather lost behind technical details of published works on this topic). Under Isabelle's pressure we eventually arrived at a~hopefully clear proof structure which includes a simple, but probably innovative use of the idea of ``gluing'' words. 
%We also propose an alternative formulation of the crucial combinatorial fact in elementary terms of primitivity and periodicity, instead of a rather specialized terminology of interpretations. This alternative formulation shows that the result can be seen as natural extension of the Periodicity lemma, which is one of the most frequently used tool in combinatorics on words. 
The analysis of the proof is therefore another, and we believe the most important contribution of our formalization, in addition to the mere certainty that there are no gaps in the proof.

In addition, we provide a complete  parametric solution of the equation $x^ky^j = z^\ell$ for arbitrary $j$, $k$ and $\ell$, a classification which is not very difficult, but maybe 
too complicated to be useful in a mere unverified paper form. 

The formalization presented here is an organic part of a larger project of formalization of combinatorics of words (see an introductory description in \cite{itp2021}). We are not aware of a similar formalization project in any proof assistant. 
The existence of the underlying library, which in turn extends the theories of ``List'' and ``HOL-Library.Sublist'' from the standard Isabelle distribution, critically contributes to a smooth formalization which is getting fairly close to the way a~human paper proof would look like, outsourcing technicalities to the (reusable) background. We accompany claims in this text with names of their formalized counterparts. 

\section{Basic facts and notation}\label{sec:notation}
Let $\Sigma$ be an arbitrary set.
Lists (i.e. finite sequences) $[x_1,x_2,\dots,x_n]$ of elements $x_i \in \Sigma$ are called \emph{words} over $\Sigma$.  
The set of all words over $\Sigma$ is usually denoted as $\Sigma^*$, using the Kleene star. A notorious ambivalence of this notation is related to the situation when we consider a set of words $X \subset \Sigma^*$, and are interested in lists over $X$. They should be denoted as elements of $X^*$. However, $X^*$ usually means something else (in the theory of rational languages), namely the set of all words in $\Sigma^*$ generated by the set $X$. To avoid the confusion, we will therefore follow the notation familiar from the formalization in Isabelle, and write $\lists X$ instead, to make clear that the entries of an element of $\lists X$ are themselves words. In order to further help to distinguish words over the basic alphabet from lists over a set of words, we shall use boldface variables for the latter. In particular, it is important to keep in mind the difference between a letter $a$ and the word $[a]$ of length one, the distinction which is usually glossed over lightly in the literature on combinatorics on words. The set of words over $\Sigma$ generated by $X$ is then denoted as $\gen X$. 
The (associative) binary operation of concatenation of two words $u$ and $v$ is denoted by $u \cdot v$ We prefer this algebraic notation to the Isabelle's original \verb|@|. Moreover, we shall often omit the dot as usual. If~$\us = [x_1,x_2,\ldots, x_n] \in \lists X$ is a list of words, then we write $\concat \us$ for $x_1\cdot x_2 \cdots x_n$.
We write $\varepsilon$ for the empty list, and $u^k$ for the concatenation of $k$ copies of $u$ (we use $u^{\verb|@|}k$ in the formalization). We  write $u \leq_p v$, $u <_p v$, $u \leq_s v$, $u <_s v$, and $u \leq_f v$ to denote that $u$ is a \emph{prefix}, a \emph{strict prefix}, \emph{suffix}, \emph{strict suffix} and \emph{factor} (that is, a contiguous sublist) respectively. 
A word is \emph{primitive} if it is nonempty and not a power of a shorter word.
Otherwise, we call it \emph{imprimitive}.
Each nonempty word $w$ is a power of a unique primitive word $\rho\, w$, its \emph{primitive root}.
A nonempty word $r$ is a \emph{periodic root} of a word $w$ if $w \leq_p r \cdot w$. This is equivalent to $w$ being a prefix of the right infinite power of $r$, denoted $r^\omega$. Note that we deal with finite words only, and we use the notation $r^\omega$ only as a convenient shortcut for ``a sufficiently long power of $r$''. 
Two words $u$ and $v$ are \emph{conjugate}, we write $u \sim v$, if $u = rq$ and $v=qr$ for some words $r$ and $q$. Note that conjugation is an equivalence whose classes are also called \emph{cyclic words}. A word $u$ is \emph{cyclic factor} of $w$ if it is a factor of some conjugate of $w$.  A set of words $X$ is a \emph{code} if its elements do not satisfy any nontrivial relation, that is, they are a basis of a free semigroup.  
For two element set $\{x,y\}$, this is equivalent to $xy\neq yx$, and/or to $\rho\, x \neq \rho\, y$.
An important characterization of a semigroup $S$ of words to be free is the \emph{stability condition} which is the implication $u,v,uz,zv \in S \Longrightarrow z \in S$. 
The longest common prefix of $u$ and $v$ is denoted by $u \wedge_p v$. If $\{x,y\}$ is a (binary) code, then $(x \cdot w) \wedge_p (y \cdot w') = xy\wedge_p yx$
for any $w,w'\in \gen{\{x,y\}}$ sufficiently long. We explain some elementary facts from combinatorics on words in more detail in the Appendix~\ref{sec:appendix}.

\section{Main theorem}
Let us introduce the central definition of the paper.
\begin{definition}
	We say that a set $X$ of words is \emph{primitivity preserving} if there is no word $\ws \in \lists X$ such that
	\begin{itemize}
		\item $\abs \ws \geq 2$;
		\item $\ws$ is primitive; and
		\item $\concat \ws$ is imprimitive.
	\end{itemize} 
\end{definition}	

Note that our definition does not take into account singletons $\ws = [x]$. In particular, $X$ can be primitivity preserving even if some $x \in X$ is imprimitive. Nevertheless, in the binary case, we will also provide some information about the cases when one or both elements of the code have to be primitive.

In \cite{Mitrana1997}, V. Mitrana formulates the primitivity of a set in terms of morphisms, and shows that $X$ is primitivity preserving if and only if it is the minimal set of generators of a ``pure monoid'', cf. \cite[p. 276]{BerstelCodes}. This brings about a wider concept of morphisms preserving a given property, most classically squarefreeness, see for example a characterization of squarefree morphisms over three letters by M. Crochemore \cite{Crochemore1982}.

The target claim of our formalization is the following characterization of words witnessing that a binary code is not primitivity preserving:

\begin{theorem}[\texttt{bin\_imprim\_code}] \label{th:Theorem1} Let $B = \{x,y\}$ be a code that is not primitivity preserving.
		Then there are integers $j \geq 1$ and $k \geq 1$, with $k = 1$ or $j = 1$, such that the following conditions are equivalent for any
		$\ws \in \lists B$ with $\abs {\ws} \geq 2$:  
		\begin{itemize}
		\item  $\ws$ is primitive, and $\concat \ws$ is imprimitive
		\item  $\ws$ is conjugate with $[x]^j[y]^k$.
		\end{itemize}
   Moreover, assuming $\abs y \leq \abs x$,
   \begin{itemize}
   	\item if $j \geq 2$, then $j=2$ and $k=1$, and both $x$ and $y$ are primitive;
   	\item if $k \geq 2$, then $j=1$ and $x$ is primitive.
   \end{itemize}
\end{theorem}

\begin{proof}
Let $\ws$ be a word witnessing that $B$ is not primitivity preserving. That is, $\abs{\ws} \geq 2$, $\ws$ is primitive, and $\concat \ws$ is imprimitive. Since $[x]^j[y]^k$ and $[y]^k[x]^j$ are conjugate, we can suppose, without loss of generality, that $\abs y \leq \abs x$.

 First, we want to show that $\ws$ is conjugate with $[x]^j[y]^k$ for some $j,k \geq 1$ such that $k = 1$ or $j = 1$.
   Since $\ws$ is primitive and of length at least two, it contains both $x$ and $y$.  If it contains one of these letters exactly once, then the desired from $[x]^j [y]^k$ up to conjugation is guaranteed, with $j = 1$ or $k = 1$.
Therefore, the difficult part is to show that no primitive $\ws$ with $\concat \ws$ imprimitive can contain both letters at least twice.
This is the main task of the rest of the paper, which is finally accomplished by Theorem \ref{th:main} claiming that words that contain at least two occurrences of $x$ are conjugate with $[x,x,y]$.
%Therefore, for each primitive $\ws \in \lists B$ of length at least two with imprimitive concatenation, we have $i$ and $j$ such that $\ws \sim [x]^j[y]^k$. 
To complete the proof of the first part of the theorem, it remains to show that $j$ and $k$ do not depend on $\ws$.
This follows from Lemma~\ref{le:jk_unique}.

Note that the imprimitivity of $\concat \ws$ induces the equality $x^jy^k = z^\ell$ for some $z$ and $\ell \geq 2$. The already mentioned seminal result of Lyndon and Sch\"utzenberger shows that $j$ and $k$ cannot be simultaneously at least two, since otherwise $x$ and $y$ commute. For the same reason, considering its primitive root, the word $y$ is primitive if $j \geq 2$.
Similarly, $x$ is primitive if $k \geq 2$. 
The primitivity of $x$ when $j = 2$ is a part of Theorem \ref{th:main}.  \qed
\end{proof}

We start by giving a complete parametric solution of the equation $x^jy^k = z^\ell$ in the following theorem.
This will eventually yield, after the proof of Theorem \ref{th:Theorem1} is completed, a full description of not primitivity preserving binary codes. Since the equation is mirror symmetric, we omit symmetric cases by assuming  $|y| \leq |x|$.

\newcommand*{\thead}[1]{\multicolumn{1}{c|}{\bfseries #1}}
\newcommand*{\theadlast}[1]{\multicolumn{1}{c}{\bfseries #1}}

\begin{theorem}[\texttt{LS\_parametric\_solution}] \label{thm:xjykzl_solution}
	Let $\ell \geq 2$, $j,k \geq 1$ and $|y| \leq |x|$.
	%  $\{x,y\}$ be a code 

	The equality $x^jy^k = z^\ell$ holds if and only if one of the following cases takes place:

	% \begin{tabular}{l|l|l|l|l}
	% \thead{condition} & \thead{$x$} & \thead{$y$} & \thead{$z$} & \\ \hline
	% $j = k = 1$ & $(rq)^tr$ & $q (rq)^{\ell-t-1}$ & $rq$ & $t < l$ \\
	% $j \geq 2, k = 1$ & $(rq)^{t+1}r$ & $qrrq$ &  $(rq)^{t+1}rrq$ & $j = l = 2$, $x$ and $y$ primitive \\
	% $j = 1, k \geq 2, y^k {<}_s z$ & $(qy^k)^{\ell-1}q$ & $y$ & $qy^k$ & $x$ primitive \\
	% $j = 1, k \geq 2, z \leq_s y^k$ & $(qr(r(qr)^{t+1})^{k - 1})^{\ell - 2} qr(r(qr)^{t+1})^{k - 2}rq$ &
 %    $r(qr)^{t+1}$ &
 %    $qr(r(qr)^{t+1})^{k - 1}$ & $x$ primitive\\ 
	% \end{tabular}

	\begin{enumerate}[A.] \setlength\itemsep{0.5em}
	\item \label{A} There exists a word $r$, and integers $m,n,t \geq 0$ such that  $$mj+nk = t \ell ,$$ and
	   \begin{align*}
		x = r^m, \quad y = r^n, \quad z = r^t; 
	   \end{align*}
	\item \label{B}  $j = k = 1$ and there exist non-commuting words $r$ and $q$, 
		and integers $m,n \geq 0$ such that $$m+n+1 = \ell ,$$ and
		$$x = (rq)^mr, \quad y = q(rq)^{n}, \quad z = rq;$$ 
	\item \label{C}  $j = \ell = 2$, $k = 1$ and there exist non-commuting words
		$r$ and $q$ and an integer $m \geq 2$ such that 
		$$x = (rq)^m r, \quad y = qrrq, \quad z = (rq)^mrrq;$$
	\item  \label{D}  $j = 1$ and $k \geq 2$ and there exist  non-commuting words
	$r$ and $q$ such that \\
		$$x = (qr^k)^{\ell-1}q, \quad y = r, \quad z = qr^k;$$
	\item \label{E}  $j = 1$ and  $k \geq 2$ and there are non-commuting words
	$r$ and $q$, an integer $m \geq 1$ such that 
		$$x = (qr(r(qr)^m)^{k - 1})^{\ell - 2}qr(r(qr)^m)^{k - 2}rq, \quad y = r(qr)^m, \quad z = qr(r(qr)^m)^{k - 1}.$$
	\end{enumerate}

	% 	\begin{tabular}{r|l|l|l|l}
	% \thead{case} & \thead{1} & \thead{2} & \thead{3} & \theadlast{4} \\
	%  & \thead{$j = k = 1$} & \thead{$j = 2, k = 1$} & \thead{$j = 1, k \geq 2,  y^k {<}_s z$} & \theadlast{$j = 1, k \geq 2, z {<}_s y^k$} \\ \hline \hline
	% $x$ & $(rq)^tr$ & $(rq)^{t+2}r$ & $(qy^k)^{\ell-1}q$ & \parbox{3cm}{\begin{flalign*}(qr(r(qr)^{t+1})^{k - 1})^{\ell - 2} \cdot &&\\ \cdot qr(r(qr)^{t+1})^{k - 2}rq &&\end{flalign*}} \\ \hline
	% $y$ & $q (rq)^{\ell-t-1}$ & $qrrq$  & $y$ & $r(qr)^{t+1}$ \\ \hline
	% $z$ & $rq$ & $(rq)^{t+2}rrq$ & $qy^k$ & $qr(r(qr)^{t+1})^{k - 1}$ \\ \hline
	%  & $t < l$ & \begin{minipage}{1.7cm}{$j = l = 2$, \\ $x$ and $y$ primitive}\end{minipage} & $x$ primitive & $x$ primitive \\
	%  \end{tabular}

	% $j = k = 1$ & $(rq)^tr$ & $q (rq)^{\ell-t-1}$ & $rq$ & $t < l$ \\
	% $j \geq 2, k = 1$ & $(rq)^{t+1}r$ & $qrrq$ &  $(rq)^{t+1}rrq$ & $j = l = 2$, $x$ and $y$ primitive \\
	% $j = 1, k \geq 2, y^k {<}_s z$ & $(qy^k)^{\ell-1}q$ & $y$ & $qy^k$ & $x$ primitive \\
	% $j = 1, k \geq 2, z \leq_s y^k$ & $(qr(r(qr)^{t+1})^{k - 1})^{\ell - 2} qr(r(qr)^{t+1})^{k - 2}rq$ &
 %    $r(qr)^{t+1}$ &
 %    $qr(r(qr)^{t+1})^{k - 1}$ & $x$ primitive\\ 
	
\end{theorem}

\begin{proof}
If $x$ and $y$ commute, then all three words commute, hence they are a power of a common word. A length argument yields the solution \ref{A}. 

Assume now that $\{x,y\}$ is a code. Then no pair of words $x$, $y$ and $z$ commutes.
We have shown in the overview of the proof of Theorem \ref{th:Theorem1} that $j = 1$ or $k = 1$
 by the Lyndon-Sch\"utzenberger theorem.
	The solution is then split into several cases.
	
\noindent	\emph{Case 1}: $j = k = 1$. \\
	Let $m$ and $r$ be such that $z^mr = x$ with $r$ a strict prefix of $z$.
	By setting $z = rq$, we obtain the solution \ref{B} with $n = \ell - m -1$.
	
\noindent	\emph{Case 2}: $j \geq 2, k = 1$.\\
%	Assume first that $x^j$ is a prefix of $z$, i.e., $x^jq = z$ for some nonempty $q$.
%	It follows that $j|x| + |y| = |x^jy| = |z^\ell| = \ell j |x| + \ell |q|$, and therefore
%	\[
%	(l-1)j|x| + \ell |q| = |y|.
%	\]
%	As $\ell \geq 2$, this contradicts $|y| \leq |x|$.
%	Therefore, $z$ is a prefix of $x^j$.
%	
	Since $|y| \leq |x|$ and $\ell \geq 2$, we have
	$$2|z| \leq |z^\ell| = |x^j| + |y| < 2|x^j|,$$
	so $z$ is a strict prefix of $x^j$.
	
	As $x^j$ has periodic roots both $z$ and $x$, and $z$ does not commute with $x$, the Periodicity lemma implies
	\[
	|x^j| < |z| + |x|.
	\]
	That is, $z = x^{j-1}u$, $x^j = zv$ and $x = uv$ for some nonempty words $u$ and $v$.
	As $v$ is a prefix of $z$, it is also a prefix of $x$. Therefore, we have
	\[
	x = uv = vu'
	\]
	for some word $u'$.
	This is a well known conjugation equality which implies $u = rq$, $u' = qr$ and $v = (rq)^nr$ for some words $r$, $q$ and an integer $n \geq 0$.
	
	We have
	\[
	j|x| + |y| = |x^jy| = |z^\ell| = \ell(j-1)|x| + \ell|u|,
	\]
	and thus $|y| = (\ell j-\ell-j)|x| + \ell|u|$.
	Since $|y| \leq |x|$, $|u| > 0$, $j \geq 2$, and $\ell \geq 2$,
	it follows that $\ell j - 	\ell - j = 0$, which implies $j = l = 2$.
	We therefore have $x^2y = z^2$ and $x^2 = zv$, hence $vy = z$.

	Combining $u = rq$, $u' = qr$, and $v = (rq)^nr$ 
	with $x = vu'$, $z = x^{j-1}u = xu = vu'u$, and $vy = z$, we obtain the solution \ref{C} with $m = n+1$. The assumption $\abs y \leq \abs x$ implies $m \geq 2$.
	
%	Replacing $x$ by $r^e$ results in $lje-l-je = 0$, which is equivalent to $je(l-1)=l$,
%	which in turn implies again $l = 2$, and hence $j = 2$ and $e = 1$.
%	In other words, $x$ is primitive.
	
%	Primitivity of $y$ follows from the Lyndon-Sch\"utzenberger theorem.
	
\noindent\emph{Case 3}: $j = 1, k \geq 2, y^k {\leq}_s z$. \\
	We have $z = qy^k$ for some word $q$.
	Noticing that $x = z^{\ell-1}q$ yields the solution \ref{D}.

\noindent\emph{Case 4}: $j = 1, k \geq 2, z <_s y^k$. \\
	This case is analogous to the second part of Case 2.
	Using the Periodicity lemma, we obtain $uy^{k-1} = z$, $y^k = vz$, and $y = vu$ with nonempty $u$ and $v$.
	As $v$ is a suffix of $z$, it is also a suffix of $y$, and we have
	\[
	y = vu = u'v
	\]
	for some $u'$.
	Plugging the solution of the last conjugation equality, namely
	$u' = rq$, $u = qr$, $v = (rq)^nr$,
	into $y = u'v$, $z = uy^{k-1}$ and $z^{\ell-1} = xv$ gives the solution \ref{E} with $m = n + 1$. 

Finally, the words $r$ and $q$ do not commute since $x$ and $y$, which are generated by $r$ and $q$, do not commute.

The proof is completed by a direct verification of the opposite implication. 	
	\qed
	
%	Primitivity of $x$ in Cases 3 and 4 follows again from the Lyndon-Sch\"utzenberger theorem.
\end{proof}

We now show that, for a given not primitivity preserving binary code, there is a unique pair of exponents $(j,k)$ such that $x^jy^k$ is imprimitive.

\begin{lemma}[\texttt{LS\_unique}] \label{le:jk_unique}
Let $B = \{x,y\}$ be a code.
Assume $j,k,j',k' \geq 1$.
If both $x^jy^k$ and $x^{j'}y^{k'}$ are imprimitive, then $j = j'$ and $k = k'$.
\end{lemma}
\begin{proof}
Let $z_1,z_2$ be primitive words and $\ell,\ell' \geq 2$ be such that
\begin{equation}\label{eq:both_imprim}
x^jy^k = z_1^\ell \quad \text{ and } \quad x^{j'}y^{k'} = z_2^{\ell'}.
\end{equation}
Since $B$ is a code, the words $x$ and $y$ do not commute.  We proceed by contradiction.
%This implies that, for $i = 1,2$, neither $y$ and $z_i$, nor $x$ and $z_i$ commute.\todo{kde se pouzilo?}
%The Lyndon-Sch\"utzenberger theorem implies that either $j$ or $k$ is one, and
%similarly either $j'$ or $k'$ is one.

\noindent \emph{Case 1}: First, assume that $j = j'$ and $k \neq k'$.\\
Let, without loss of generality, $k < k'$.
From \eqref{eq:both_imprim} we obtain $z_1^\ell y^{k' - k} = z_2^{\ell'}$.
The case $k' - k \geq 2$ is impossible due to the Lyndon-Sch\"utzenberger theorem.
Hence $k' - k = 1$. This is another place where the formalization triggered a simple and nice general lemma (easily provable by the Periodicity lemma) which will turn out to be useful also in the proof of Theorem \ref{th:main}. Namely, the lemma \verb|imprim_ext_suf_comm| claims that if both $uv$, and $uvv$ are imprimitive, then $u$ and $v$ commute. We apply this lemma to $u = x^jy^{k-1}$ and $v = y$, obtaining a contradiction with the assumption that $x$ and $y$ do not commute.
%
% and $z_1^\ell y = z_2^{\ell'}$.
%Consider the word $y z_1^\ell$ which is conjugate to $z_2^\ell$, hence it is imprimitive.
%But the word $y z_1^\ell$ is also a suffix of $z_1^\ell z_1^\ell$, which implies, by the Periodicity lemma, that the primitive roots of $y z_1^\ell$ and $z_1^\ell$ are the same. Then $z_1$ and $y$ commute, a contradiction.

\noindent \emph{Case 2.} The case $k = k'$ and $j \neq j'$ is symmetric to Case 1.

\noindent \emph{Case 3.} Let finally $j \neq j'$ and $k \neq k'$. 
The Lyndon-Sch\"utzenberger theorem implies that either $j$ or $k$ is one, and
similarly either $j'$ or $k'$ is one. We can therefore assume that $k = j' = 1$ and $k',j \geq 2$. Moreover, we can assume that $\abs y \leq \abs x$. Indeed, in the opposite case, we can consider the words  $y^kx^j$ and $y^{k'}x^{j'}$ instead, which are also both imprimitive.

Theorem~\ref{thm:xjykzl_solution} now allows only the case \ref{C} for the equality $x^jy = z_1^\ell$. We therefore have $j = \ell = 2$ and 
$x = (rq)^mr$, $y = qrrq$ for an integer $m \geq 2$ and some non-commuting words $r$ and $q$.
Since $y = qrrq$ is a suffix of $z_2^\ell$, this implies that $z_2$ and $rq$ do not commute. Consider the word $x \cdot qr = (rq)^mrqr$, which is a prefix of $xy$, and therefore also of $z_2^\ell$. This means that $x \cdot qr$ has two periodic roots, namely $rq$ and $z_2$, and the Periodicity lemma implies that $\abs {x \cdot qr}< \abs{rq}  + \abs {z_2}$. Hence $x$ is shorter than $z_2$. The equality $xy^{k'} = z_2^{\ell'}$, with $\ell' \geq 2$, now implies  on one hand that $rqrq$ is a prefix of $z_2$, and on the other hand that $z_2$ is a suffix of $y^{k'}$.
It follows that $rqrq$ is a factor of $(qrrq)^k$. 
Hence $rqrq$ and $qrrq$ are conjugate, which is possible only if $r$ and $q$ commute, a fact not difficult to prove, see Appendix~\ref{sec:appendix}. This is a contradiction. \qed
\end{proof}

The rest of the paper, and therefore also of the proof of Theorem~\ref{th:Theorem1}, is organized as follows. 
In Section~\ref{sec:interpretace}, we introduce a general theory of interpretations, which is behind the main idea of the proof, and apply it to the (relatively simple) case of a binary code with words of the same length. In Section~\ref{sec:square} we characterize the unique disjoint extendable $\{x,y\}$-interpretation of the square of the longer word $x$.
This is a result of independent interest, and also the cornerstone of the proof of Theorem~\ref{th:Theorem1} which is completed in Section~\ref{sec:main} by showing that a word containing at least two $x$'s witnessing that $\{x,y\}$ is not primitivity preserving is conjugate with $[x,x,y]$.

\section{Interpretations and the main idea}\label{sec:interpretace}
Let $X$ be a code, let $u$ be a factor of $\concat \ws$ for some $\ws \in \lists X$. The natural question is to decide how $u$ can be produced as a factor of words from $X$, or, in other words, how it can be interpreted in terms of $X$. This motivates the following definition. 
\begin{definition} \label{def:interpretation} Let $X$ be a set of words over $\Sigma$. We say that the triple $(p,s,\ws) \in \Sigma^*\times \Sigma^* \times \lists X$ is an \emph{$X$-interpretation} of a word $u \in \Sigma^*$ if
\begin{itemize}
	\item 	$\ws$ is nonempty;
	\item 	$p \cdot u \cdot s = \concat \ws$;
	\item  $p <_p \hd \ws$ and
	\item  $s <_s \last \ws$.
\end{itemize}	
\end{definition}
The definition is illustrated by the following figure, where $\ws = [w_1,w_2,w_3,w_4]$:
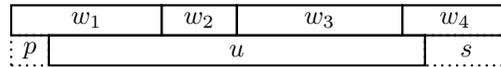
\[
\begin{tikzpicture}[thick]
	\draw (-2.5,-.2) rectangle (2.5,.2); \node at (0,0) {$u$};
	\draw (-3,.2) rectangle (-1,.6); \node at (-2,.4) {$w_1$};
	\draw (-1,.2) rectangle (0,.6); \node at (-.5,.4) {$w_2$};
	\draw (0,.2) rectangle (2.2,.6); \node at (1.1,.4) {$w_3$};
	\draw (2.2,.2) rectangle (3.6,.6); \node at (2.9,.4) {$w_4$};
	\draw[dotted] (-3,-.2) rectangle (-2.5,.2); \node at (-2.75,0) {$p$};
	\draw[dotted] (2.5,-.2) rectangle (3.6,.2); \node at (3.05,0) {$s$};
\end{tikzpicture}
\]
The first condition of the definition motivates the notation $p\, u\, s \sim_{\mathcal I} \ws$ for the situation when $(p, s, \ws)$ is an $X$-interpretation of $u$.

\begin{remark}
For sake of historical reference, we remark that our definition of $X$-interpretation differs from the one used in \cite{lerest}. Their formulation of the situation depicted by the above figure would be that $u$ is interpreted by the triple $(s', w_2 \cdot w_3, p')$ where $p\cdot s' = w_1$ and $p'\cdot s = w_4$.
This is less convenient for two reasons. First, the decomposition of $w_2 \cdot w_3$ into $[w_2,w_3]$ is only implicit here (and even ambiguous if $X$ is not a code). Second, while it is required that the the words $p'$ and $s'$ are a prefix and a suffix, respectively, of an element from $X$, the identity of that element is left open, and has to be specified separately.
\end{remark}

If $u$ is a nonempty element of $\gen X$ and $u = \concat \us$ for $\us \in \lists X$, then the $X$-interpretation $\emp\, u\, \emp \sim_{\mathcal I} \us$ is called \emph{trivial}.
Note that the trivial $X$-interpretation is unique if $X$ is a code. 

As nontrivial $X$-interpretations of elements from $\gen X$ are of particular interest, the following two concepts are useful.
\begin{definition}
	An $X$-interpretation $p\, u\, s \sim_{\mathcal I} \ws$ of $u = \concat \us$ is called 
	\begin{itemize}
		\item \emph{disjoint} if $\concat \ws' \neq p \cdot \concat \us'$ whenever $\ws' \leq_p \ws$ and $\us' \leq_p \us$.
		\item \emph{extendable} if $p \leq_s w_p$ and $s \leq_p w_s$ for some elements $w_p, w_s \in \gen X$.
	\end{itemize}
\end{definition}

Note that a disjoint $X$-interpretation is not trivial, and that being disjoint is relative to a chosen factorization $\us$ of $u$ (which is nevertheless unique if $X$ is a code).
\medskip

The definitions above are naturally motivated by \textbf{the main idea} of the characterization of sets $X$ that do not preserve primitivity, which dates back to Lentin and Sch\"utzenberger \cite{lentin}. If $\ws$ is primitive, while $\concat \ws$ is imprimitive,
say $\concat \ws = z^k$, $k\geq 2$, then the shift by $z$ provides a nontrivial and extendable $X$-interpretation of $\concat \ws$. (In fact, $k-1$ such nontrivial interpretations). Moreover, the following lemma, formulated in a more general setting of two words $\ws_1$ and $\ws_2$, implies that the $X$-interpretation is disjoint if $X$ is a code.
\begin{lemma}[\texttt{shift\_interpret}, \texttt{shift\_disjoint}]\label{lem:disjoint_interp}
	Let $X$ be a code. Let $\ws_1, \ws_2 \in \lists X$ be such that $z \cdot \concat \ws_1 = \concat \ws_2 \cdot z$ where $z \notin \gen X$. Then 
	$z \cdot \concat \vs_1 \neq \concat \vs_2$, whenever $\vs_1 \leq_p \ws_1^n$ and $\vs_2 \leq_p \ws_2^n$, $n\in \mathbb N$. 
	
	In particular $\concat \us$ has a disjoint extendable $X$-interpretation for any prefix $\us$ of $\ws_1$.  
\end{lemma}
The excluded possibility is illustrated by the following figure.
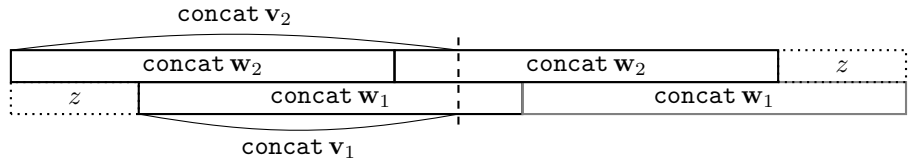
\[
\begin{tikzpicture}[thick, scale = .85]
	\draw (2,-.5) rectangle (8,0); \node at (5,-.25) {$\concat \ws_1$};
	\draw[gray] (8,-.5) rectangle (14,0); \node at (11,-.25) {$\concat \ws_1$};
	\draw (0,0) rectangle (6,.5); \node at (3,.25) {$\concat \ws_2$};
    \draw (6,0) rectangle (12,.5); \node at (9,.25) {$\concat \ws_2$};	
    \draw[dotted] (0,-.5) rectangle (2,0); \node at (1,-.25) {$z$};	
    \draw[dotted] (12,0) rectangle (14,.5); \node at (13,.25) {$z$};	
    \draw [dashed] (7,.7) -- (7,-.7);
    \draw [thin] (2,-.5) to[bend right = 10] node[below] {$\concat \vs_1$} (7, -.5);  
    \draw [thin] (0,.5) to[bend left = 7] node[above] {$\concat \vs_2$} (7, .5);  
\end{tikzpicture}
\]
\begin{proof}
	First, note that $z \cdot \concat \ws_1^n = \concat \ws_2^n \cdot z$ for any $k$.  Let $\ws_1^n = \vs_1\cdot \vs_1'$ and 
	 $\ws_2^n= \vs_2\cdot \vs_2'$.	If $z \cdot \concat \vs_1 = \concat \vs_2$, then also $\concat \vs_2' \cdot z = \concat \vs_1'$.
	 This contradicts $z \notin \gen X$ by the stability condition. 
	 
	 An extendable $X$-interpretation of $\us$ is induced by the fact that $\concat \us$ is covered by $\concat (\ws_2 \cdot \ws_2)$. 
	 	 The interpretation is disjoint by the first part of the proof. \qed
\end{proof}
In order to apply the above lemma to the imprimitive $\concat \ws = z^k$ of a primitive $\ws$, set $\ws_1= \ws_2 = \ws$. The assumption $z \notin \gen X$ follows from the primitivity of $\ws$: indeed, if $z = \concat \zs$, with $\zs \in \lists X$, then $\ws = \zs^k$ since $B$ is a code. 

%\section{Interpreting conjugate words}
% uniform_square_interp
We first apply the main idea to a relatively simple case of nontrivial $\{x,y\}$-interpretation of the word $x \cdot y$ where $x$ and $y$ are of the same length.
\begin{lemma}[\texttt{uniform\_square\_interp}]\label{lem:xy_interp}
	Let $B = \{x,y\}$ be a code with $\abs x = \abs y$. Let $p\ (x\cdot y)\ s \sim_{\mathcal I} \vs$ be a nontrivial $B$-interpretation.
	Then $\vs = [x,y,x]$ or $\vs = [y,x,y]$ and $x\cdot y$ is imprimitive.
\end{lemma}
\begin{proof}
From $p \cdot x \cdot y \cdot s = \concat \vs$, it follows, by a length argument, that $\abs \vs$ is three. A straightforward way to proof the claim is to consider all eight possible candidates.
In each case, it is then a routine few line proof that shows that $x = y$, unless $\vs = [x,y,x]$ or $\vs = [y,x,y]$, which we omit. In the latter cases, $x \cdot y$ is a nontrivial factor of its square $(x \cdot y)\cdot (x \cdot y)$, which yields the imprimitivity of $x \cdot y$.  \qed
\end{proof}
The previous (sketch of the) proof nicely illustrates  on a small scale the advantages of formalization. It is not necessary to choose between a tedious elementary proof for sake of completeness on one hand, and the suspicion that something was missed on the other hand (leaving aside that the same suspicion typically remains even after the tedious proof). A bit ironically, the most difficult part of the formalization is to show that $\vs$ is indeed of length three, which needs no further justification in a human proof.

We have the following corollary which is a variant of Theorem \ref{th:main}, and also illustrates the main idea of its proof.
\begin{lemma}[\texttt{bin\_imprim\_not\_conjug}]\label{lem:not_conjugate}
	Let $B = \{x,y\}$ be a binary code with $\abs x = \abs y$. 
If $\ws \in \lists B$ is such that $\abs {\ws} \geq 2$, $\ws$ is primitive, and $\concat \ws$ is imprimitive, 
%then $\ws$ is $[x,y]$ or $[y,x]$. 
then $x$ and $y$ are not conjugate.	
\end{lemma}
\begin{proof}
	Since $\ws$ is primitive and of length at least two, it contains both letters $x$ and $y$. Therefore, it has either $[x,y]$ or $[y,x]$ as a factor. The imprimitivity of $\concat \ws$ yields a nontrivial $B$-interpretation of $x \cdot y$, which implies that $x \cdot y$ is not primitive by Lemma \ref{lem:xy_interp}.
		 
    Let $x$ and $y$ be conjugate, and let $x = r \cdot q$ and $y = q \cdot r$. Since $x \cdot y = r \cdot q \cdot q \cdot r$ is imprimitive, also 
     $r \cdot r \cdot q \cdot q$ is imprimitive. Then $r$ and $q$ commute by the theorem of Lyndon and Sch\"utzenberger, a contradiction with $x \neq y$. \qed
\end{proof}

\section{Binary interpretation of a square}\label{sec:square}
Let $B = \{x,y\}$ be a code such that $\abs y \leq \abs x$.
In accordance with the main idea, the core technical component of the proof is the description of the disjoint extendable $B$-interpretations of the square $x^2$. 
This is a very nice result which is relatively simple to state but difficult to prove, and which is valuable on its own. As we mentioned already, it can be obtained from Th\'eor\`eme 2.1 and Lemme 3.1 in \cite{lerest}.

\begin{theorem}[\texttt{square\_interp\_ext.sq\_ext\_interp}] \label{thm:sq_interp}
Let $B = \{x,y\}$ be a code such that  $\abs y \leq \abs x$, both $x$ and $y$ are primitive, and
	$x$ and $y$ are not conjugate.

Let $p\, (x\cdot x)\, s \sim_{\mathcal I} \ws$ be a disjoint extendable $B$-interpretation. 
Then
\begin{align*}
 \ws &= [x,y,x], &
  s \cdot p &= y, &
p \cdot x &= x \cdot s.
\end{align*}
\end{theorem}

In order to appreciate the theorem, note that the definition of interpretation implies
\[p \cdot x \cdot x \cdot s = x \cdot y \cdot x,\] 
hence $x \cdot y \cdot x = (p \cdot x)^2$. This will turn out to be the only way how primitivity may not be preserved if $x$ occurs at least twice in $\ws$. Here is an example with $x = \verb|01010|$ and $y = \verb|1001|$:
\[
\begin{tikzpicture}[thick]
	\foreach \x/\p in {0/{\tt 0},1/{\tt 1},2/{\tt 0},3/{\tt 1},4/{\tt 0},5/{\tt 1},6/{\tt 0},7/{\tt 0},8/{\tt 1},9/{\tt 0},10/{\tt 1},11/{\tt 0},12/{\tt 1},13/{\tt 0}} 
	\node (p\x) at (\x*.45,0) {\p};
	\draw (p0.south west) rectangle (p4.north east); 
	\draw (p4.south east) rectangle (p8.north east); 
	\draw (p8.south east) rectangle (p13.north east); 
	\foreach \x/\p in {0/{\tt 0},1/{\tt 1},2/{\tt 0},3/{\tt 1},4/{\tt 0},5/{\tt 1},6/{\tt 0},7/{\tt 0},8/{\tt 1},9/{\tt 0},10/{\tt 1},11/{\tt 0},12/{\tt 1},13/{\tt 0}} 
     \node (q\x) at (\x*.45,-.45) {\p};
\draw (p2.south west) rectangle (q6.south east); 
\draw (p6.south east) rectangle (q11.south east); 
\draw[dotted] (p0.south west) rectangle (q2.south west);
\draw[dotted] (p11.south east) rectangle (q13.south east);
%\draw (p9.south west) rectangle (p13.north east); 

%	\draw (-2.5,-.2) rectangle (2.5,.2); \node at (0,0) {$u$};
%	\draw (-3,.2) rectangle (-1,.6); \node at (-2,.4) {$w_1$};
%	\draw (-1,.2) rectangle (0,.6); \node at (-.5,.4) {$w_2$};
%	\draw (0,.2) rectangle (2.2,.6); \node at (1.1,.4) {$w_3$};
%	\draw (2.2,.2) rectangle (3.6,.6); \node at (2.9,.4) {$w_4$};
%	\draw[dotted] (-3,-.2) rectangle (-2.5,.2); \node at (-2.75,0) {$p$};
%	\draw[dotted] (2.5,-.2) rectangle (3.6,.2); \node at (3.05,0) {$s$};
\end{tikzpicture}
\]

\begin{proof}
By the definition of a disjoint interpretation, we have $p\cdot x\cdot x \cdot s = \concat \ws$, where $p \neq \varepsilon$	and $s \neq \varepsilon$. A length argument implies that $\ws$ has length at least three. 
Since a primitive word is not a nontrivial factor of its square, we have $\ws = [\hd \ws] \cdot [y]^k \cdot [\last \ws]$, with $k \geq 1$. 
Since the interpretation is disjoint, we can split the equality into $p \cdot x = \hd \ws \cdot y^m \cdot u$ and $x \cdot s = v \cdot y^\ell \cdot \last \ws$, where $y = u \cdot v$, both $u$ and $v$ are nonempty, and $k = \ell + m + 1$.  
 We want to show $\hd \ws = \last \ws = x$ and $m = \ell = 0$.
 The situation is mirror symmetric so we can solve cases two at a time.
 
If $\hd \ws = \last \ws = y$, then powers of $x$ and $y$ share a factor of length at least $\abs x + \abs y$. Since they are primitive, this implies that they are conjugate, a contradiction. The same argument applies when 
 $\ell \geq 1$ and $\hd \ws = y$  (if $m \geq 1$ and $\last \ws = y$
 respectively). Therefore, in order to prove $\hd \ws = \last \ws = x$, it remains to exclude the case $\hd \ws = y$, $\ell = 0$ and $\last \ws = x$ ($\last \ws = y$, $m = 0$ and $\hd \ws = x$ respectively).
% This is excluded by one of the technical lemmas that we single out, since it implies that $y$ is not primitive:
This is covered by one of the technical lemmas that we single out:

\begin{lemma}[\texttt{pref\_suf\_pers\_short}]\label{lem:aux1}
Let $x \leq_p v \cdot x$, $x \leq_s p \cdot u \cdot v \cdot u$	and $\abs x > \abs{v \cdot u}$ with $p \in \gen{\{u,v\}}$.
Then $u \cdot v = v \cdot u$.
\end{lemma}

This lemma indeed excludes the case we wanted to exclude, since the conclusion implies that $y$ is not primitive. 
We skip the proof of the lemma here and make instead an informal comment. Note that $v$ is a period root of $x$. In other words, $x$ is a factor of $v^\omega$. Therefore, with the stronger assumption that
$v \cdot u \cdot v$ is a factor of $x$, the conclusion follows easily by the familiar principle that $v$ being a factor of $v^\omega$ ``synchronizes'' primitive roots of $v$.   	 
Lemma~\ref{lem:aux1} then exemplifies one of the virtues of formalization, which makes it easy to generalize auxiliary lemmas, often just by following the most natural proof and checking its minimal necessary assumptions.

Now we have $\hd \ws = \last \ws = x$, hence  $p \cdot x = x \cdot y^m \cdot u$ and $x \cdot s = v \cdot y^\ell \cdot x$. The natural way to describe this scenario is to observe that $x$ has both the (prefix) period root $v \cdot y^\ell$, and the suffix period root $y^m \cdot u$.
Using again Lemma~\ref{lem:aux1}, we exclude situations when $\ell = 0$ and $m \geq 1$ ($m = 0$ and $\ell \geq 1$ resp.).
It therefore remains to deal with the case when both $m$ and $\ell$ are positive. We divide this into four lemmas according to the size of the overlap the prefix $v\cdot y^\ell$ and the suffix $y^m\cdot u$ have in $x$. More exactly, the cases are:
\begin{itemize}
	\item $\abs{v\cdot y^\ell} + \abs{y^m\cdot u} \leq \abs x$
	\item $\abs x < \abs{v\cdot y^\ell} + \abs{y^m\cdot u} \leq \abs x + \abs u$
	\item $\abs x + \abs u < \abs{v\cdot y^\ell} + \abs{y^m\cdot u} < \abs x + \abs {u \cdot v}$
	\item $\abs x + \abs {u\cdot v} \leq \abs{v\cdot y^\ell} + \abs{y^m\cdot u}$
\end{itemize} 
and they are solved by an auxiliary lemma each. The first three cases yield that $u$ and $v$ commute, the first one being a straightforward application of the Periodicity lemma. The last one is also straightforward application of the ``synchronization'' idea. It implies that $x \cdot x$ is a factor of $y^\omega$, a contradiction with the assumption that $x$ and $y$ are primitive and not conjugate.
Consequently, the technical, tedious part of the whole proof is concentrated in lemmas dealing with the second, and the third case (see lemmas {\verb|short_overlap|} and {\verb|medium_overlap|} in the theory \verb|Binary_Square_Interpretation.thy|). The corresponding proofs are further analyzed and decomposed into more elementary claims in the formalization, where further details can be found. 

This completes the proof of $\ws = [x,y,x]$. A byproduct of the proof is the description of words $x$, $y$, $p$ and $s$. Namely, there are non-commuting words $r$ and $t$, and integers $m$, $k$ and $\ell$ such that
\begin{align*}
	x &= (rt)^{m+1}\cdot r, & y &= (tr)^{k+1}\cdot (rt)^{\ell+1}, & p &= (rt)^{k+1}, & s &= (tr)^{\ell+1}\,.
\end{align*}
%
%\begin{itemize}
%	\item $x = (rt)^{m+1}\cdot r$
%	\item $y = (tr)^{k+1}\cdot (rt)^{\ell+1}$
%	\item $p = (rt)^{k+1}$
%	\item $s = (tr)^{\ell+1}$
%\end{itemize}  
The second claim of the present theorem, that is, $y = s \cdot p$ is then equivalent to $k = \ell$, and it is an easy consequence of the assumption that the interpretation is extendable.  \qed
 
\end{proof}

\section{The witness with two $x$'s} \label{sec:main}

In this section, we characterize 
words witnessing that $\{x,y\}$ is not primitivity preserving and containing at least two $x$'s.

\begin{theorem}[\texttt{bin\_imprim\_longer\_twice}]\label{th:main}
	Let $B = \{x,y\}$ be a code such that $\abs y \leq \abs x$. Let $\ws \in \lists \{x,y\}$ be a primitive word which contains $x$ at least twice 
	such that $\concat \ws$ is imprimitive. 
	
	Then $\ws \sim [x,x,y]$ and both $x$ and $y$ are primitive.
\end{theorem}

We divide the proof in three steps.

\subsubsection{The core case.}
We first prove the claim with two additional assumptions which will be subsequently removed. Namely, the following lemma shows how the knowledge about the $B$-interpretation 
of $x \cdot x$ from the previous section is used. The additional assumptions are displayed as items.
\begin{lemma}[\texttt{bin\_imprim\_primitive}] \label{lem:all_primitive}
	Let $B = \{x,y\}$ be a code with $\abs y \leq \abs x$ where
	\begin{itemize}
		\item both $x$ and $y$ are primitive,
	\end{itemize}
  and let $\ws \in \lists B$ be primitive such that $\concat \ws$ is imprimitive, and 
  \begin{itemize}
  	\item $[x,x]$ is a cyclic factor of  $\ws$.
  \end{itemize}
 Then $\ws \sim [x,x,y]$.  
\end{lemma}
\begin{proof}
	Choosing a suitable conjugate of $\ws$, we can suppose, without loss of generality, that $[x,x]$ is a prefix of $\ws$.
	Now, we want to show $\ws = [x,x,y]$.
	Proceed by contradiction and assume $\ws \neq [x,x,y]$.
	Since $\ws$ is primitive, this implies $\ws \cdot [x,x,y] \neq [x,x,y] \cdot \ws$.

	By Lemma~\ref{lem:not_conjugate}, we know that $x$ and $y$ are not conjugate.
	Let $\concat \ws = z^k$, $2 \leq k$ and $z$ primitive.
Lemma~\ref{lem:disjoint_interp} yields a disjoint extendable $B$-interpretation of $(\concat \ws)^2$. In particular, the induced disjoint extendable $B$-interpretation of the prefix $x \cdot x$ is of the form $p\, (x \cdot x)\, s \sim_{\mathcal I} [x,y,x]$ by Theorem~\ref{thm:sq_interp}:
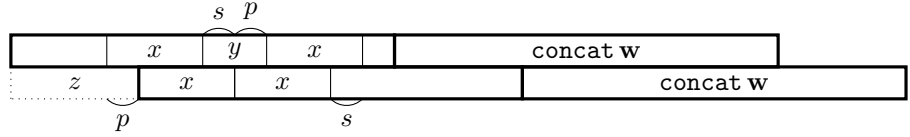
\[
\begin{tikzpicture}[very thick, scale = .85]
	\draw[dotted,thin] (0,-.5) rectangle (2,0); \node at (1,-.25) {$z$};
	\draw (2,-.5) rectangle (8,0); \node at (5,-.25) {};
	\draw[thin] (2,-.5) rectangle (3.5,0); \node at (2.75,-.25) {$x$};
	\draw[thin] (3.5,-.5) rectangle (5,0); \node at (4.25,-.25) {$x$};
	\draw[thin] (1.5,0) rectangle (3,.5); \node at (2.25,.25) {$x$};
	\draw[thin] (4,0) rectangle (5.5,.5); \node at (4.75,.25) {$x$};
	\node at (3.5,.25) {$y$};
	\draw (0,0) rectangle (6,.5); \node at (3,.25) {};
    \draw (6,0) rectangle (12,.5); \node at (9,.25) {$\concat \ws$};
    \draw (8,-.5) rectangle (14,.0); \node at (11,-.25) {$\concat \ws$};
    \draw [thin] (1.5,-.5) to[bend right = 45] node[below] {$p$} (2, -.5); 
    \draw [thin] (5,-.5) to[bend right = 45] node[below] {$s$} (5.5, -.5);     
     \draw [thin] (3,.5) to[bend left = 45] node[above] {$s$} (3.5, .5); 
     \draw [thin] (3.5,.5) to[bend left = 45] node[above] {$p$} (4, .5);   
\end{tikzpicture}
\]
Let $\ps$ be the prefix of $\ws$ such that $\concat \ps \cdot p = z$. Then 
\begin{equation*}%\label{eq:concat}
	\concat (\ps \cdot [x,y])  = z \cdot (x \cdot p),  \quad	\concat [x,x,y] = (x \cdot p)^2,  \quad
	\concat \ws = z^k, 
\end{equation*}
and we want to show $z = xp$, which will imply $\concat ([x,x,y]\cdot \ws) = \concat (\ws \cdot [x,x,y])$, hence $\ws = [x,x,y]$ since $\{x,y\}$ is a code, and both $\ws$ and $[x,x,y]$ are primitive.
 
Again, proceed by contradiction, and assume $z \neq xp$. Then, since both $z$ and $x\cdot p$ are primitive, they do not commute. We now have two binary codes, namely $\{\ws,[x,x,y]\}$ and $\{z,xp\}$. The following two equalities, \eqref{eq:lcp1} and \eqref{eq:lcp2} exploit the fundamental property of longest common prefixes of elements of binary codes mentioned in Section \ref{sec:notation}. In particular, we need the following lemma:

\begin{lemma}[\texttt{bin\_code\_lcp\_concat}]\label{lem:lcp}
Let $X = \{u_0,u_1\}$ be a binary code, and let $\zs_0,\zs_1 \in \lists X$ be such that $\concat \zs_0$ and $\concat \zs_1$ are not prefix-comparable. Then
\[(\concat \zs_0) \wedge_p (\concat \zs_1) = \concat (\zs_0 \wedge_p \zs_1) \cdot (u_0 \wedge u_1).\]
\end{lemma}
See Appendix \ref{sec:appendix} for more comments on this property.
Denote $\alpha_{z,xp} = z \cdot xp \wedge_p xp \cdot z$. 
Then also $\alpha_{z,xp} = z^k \cdot (xp)^2 \wedge_p (xp)^2 \cdot z^k$.
Similarly, let $\alpha_{x,y} = x \cdot y \wedge_p y \cdot x$. Then Lemma \ref{lem:lcp} yields
\begin{equation}
\begin{aligned} \label{eq:lcp1}
	\alpha_{z,xp} &= \concat (\ws \cdot [x,x,y]) \wedge_p \concat ([x,x,y] \cdot \ws) = \\
	&= \concat (\ws \cdot [x,x,y] \wedge_p  [x,x,y] \cdot \ws) \cdot \alpha_{x,y}
\end{aligned}
\end{equation}
and also
\begin{equation}
	\begin{aligned}\label{eq:lcp2}
	z \cdot \alpha_{z,xp} = &\concat (\ws \cdot \ps \cdot [x,y]) \wedge_p \concat (\ps\cdot[x,y] \cdot \ws) = \\
	= &\concat (\ws \cdot \ps \cdot [x,y] \wedge_p \ps\cdot[x,y] \cdot \ws) \cdot \alpha_{x,y}.
\end{aligned}
\end{equation}
Denote
\begin{align*}
	\vs_1 &= \ws \cdot [x,x,y] \wedge_p [x,x,y] \cdot \ws,  &	\vs_2 &= \ws \cdot \ps \cdot [x,y] \wedge_p \ps\cdot[x,y] \cdot \ws.
\end{align*} 
%Observe that both $\ws_1$ and $\ws_2$ are prefixes of $\ws^3$, and let
%$\ws_1\cdot\ws_1' = \ws_2\cdot\ws_2' = \ws^3$.
From \eqref{eq:lcp1} and \eqref{eq:lcp2} we now have
%\begin{align*}
$	z \cdot \concat \vs_1 = \concat \vs_2$. Since $\vs_1$ and $\vs_2$ are prefixes of some $\ws^n$, we have 
%\end{align*}
a contradiction with Lemma \ref{lem:disjoint_interp}. 
%\begin{align*}
%	\concat \ws_1' \cdot \concat \ws_1 &= \concat \ws_2'\cdot \concat \ws_2,
%\end{align*}
%and $\ws_1'\cdot\ws_1 = \ws_2'\cdot\ws_2$, since $B$ is a code.
%From $\ws_1\cdot\ws_1' = \ws_2\cdot\ws_2' = \ws^2$ and $\ws_1 = \concat \ws_2$ we obtain 
%a word $\vs \in \lists B$ such that $\concat vs = z$, 
%$\ws_2 = \ws_1\cdot \vs$, $\vs \cdot \ws_2' = \ws_1'$ and $\abs{\concat \vs} = \abs z$. The equality  $\ws_1'\cdot\ws_1 = \ws_2'\cdot\ws_2$ now becomes  $\vs\cdot \ws_2'\cdot\ws_1 = \ws_2'\cdot\ws_1 \cdot \vs$ which implies that $\ws^2$ is a power of a word conjugate with $\vs$, 
%a contradiction with the primitivity of $\ws$. 
\qed
\end{proof}

\subsubsection{Dropping the primitivity assumption.} 
We first deal with the situation when $x$ and $y$ are not primitive. A natural idea is to consider the primitive roots of $x$ and $y$ instead of $x$ and $y$. This means that we replace the word 
$\ws$ with $\R \ws$, where $\R$ is the morphism mapping $[x]$ to $[\rho\, x]^{e_x}$ and $[y]$ to $[\rho\, y]^{e_y}$ where $x = (\rho\, x)^{e_x}$ and $y = (\rho\, y)^{e_y}$. For example, if $x = abab$ and $y = aa$, and $\ws = [x,y,x] = [abab,aa,abab]$, then $\R \ws = [ab,ab,a,a,ab,ab]$. 

Let us check which hypotheses of Lemma \ref{lem:all_primitive} are satisfied in the new setting, that is, for the code $\{\rho\,x,\rho\,y\}$ and the word $\R \ws$. 
The following facts are not difficult to see.
	\begin{itemize}
		\item $\concat \ws = \concat (\R \ws)$;
%		\item if $x$ or $y$ is imprimitive, then $\abs{\rho\,x} + \abs{\rho\,y} < \abs x + \abs y$;
		\item if $[c,c]$, $c\in \{x,y\}$, is a cyclic factor $\ws$, then $[\rho\,c,\rho\,c]$ is a cyclic factor of $\R \ws$.
	\end{itemize}
The next required property: 
\begin{itemize}
	\item if $\ws$ is primitive, then $\R \ws$ is primitive;
\end{itemize}
deserves more attention. It triggered another little theory of our formalization which can be found in locale \verb|sings_code|. Note that it fits well into our context, since the claim is that $\R$ is a primitivity preserving morphism, which implies that its image on the singletons $[x]$ and $[y]$ forms a primitivity preserving set of words, see theorem \verb|code.roots_prim_morph|.

%\begin{remark}
%With respect to the last item above, note that although we assume in Lemma \ref{lem:all_primitive} that $[x,x]$ is a prefix of $\ws$ for the sake of simplicity, the crucial assumption is that it is a cyclic factor. We can always consider the conjugate where a chosen factor is a prefix, since the primitivity is a property of the conjugacy class, rather then an individual word. The conclusion of 
%\end{remark}

Consequently, the only missing hypothesis preventing the use of Lemma \ref{lem:all_primitive} is $\abs y \leq \abs x$ since it may happen that $\abs {\rho\, x} < \abs{\rho\,y}$. In order to solve this difficulty, we shall ignore for a while the length difference between $x$ and $y$, and obtain the following intermediate lemma.

\begin{lemma}[\texttt{bin\_imprim\_both\_squares}, \texttt{bin\_imprim\_both\_squares\_prim}]\label{lem:two_squares}
	Let $B = \{x,y\}$ be a code, and let $\ws \in \lists B$ be a primitive word such that $\concat \ws$ is imprimitive.
	Then $\ws$ cannot contain both $[x,x]$ and $[y,y]$ as cyclic factors.
\end{lemma}
\begin{proof}
	Assume that $\ws$ contains both $[x,x]$ and $[y,y]$ as cyclic factors.
	
	Consider the word $\R \ws$ and the code $\{\rho\, x,\rho\, y\}$. Since $\R \ws$ contains both $[\rho\,x,\rho\,x]$ and $[\rho\,y,\rho\,y]$, Lemma \ref{lem:all_primitive}
	implies that $\R \ws$ is conjugate either with the word $[\rho\,x,\rho\,x,\rho\,y]$ or with $[\rho\,y,\rho\,y,\rho\,x]$, which is a contradiction with the assumed presence of both squares.  
  \qed
\end{proof}

\subsubsection{Concluding the proof by gluing.}\label{sec:gluing}
It remains to deal with the existence of squares. We use an idea that is our main innovation with respect to the proof from \cite{lerest}, and contributes significantly to the reduction of length of the proof, and hopefully also to its increased clarity. 
Let $\ws$ be a list over a set of words $X$. The idea is to choose one of the words, say $u \in X$, and 
to concatenate (or ``glue'') blocks of $u$'s to words following them. For example, if $\ws = [u,v,u,u,z,u,z]$, then the resulting list is $[uv,uuz,uz]$. This procedure is in the general case well defined on lists whose last ``letter'' is not the chosen one and it leads to a new alphabet $\{u^i\cdot v \mid v \neq u\}$ which is a code if and only if $X$ is. This idea is used in an elegant proof of the Graph lemma (see \cite{itp2021} and \cite{Berstel1979}). In the binary case, which is of interest here, if $\ws$ in addition does not contain a square of a letter, say $[x,x]$, then 
the new code $\{x\cdot y, y\}$ is again binary. Moreover, the resulting glued list $\ws'$ has the same concatenation, and it is primitive if (and only if) $\ws$ is. Note that gluing is in this case closely related to the Nielsen transformation $y \mapsto x^{-1}y$ known from the theory of automorphisms of free groups.

Induction on $\abs \ws$ now easily leads to the proof of Theorem \ref{th:main}.
\begin{proof}[of Theorem \ref{th:main}]
	If $\ws$ contains $y$ at most once, then we are left with the equation $x^j\cdot y = z^\ell$, $\ell \geq 2$. The equality $j = 2$ follows from the Periodicity lemma, see Case 2 in the proof of Theorem \ref{thm:xjykzl_solution}.
	
	Assume for contradiction that $y$ occurs at least twice in $\ws$. Lemma \ref{lem:two_squares} implies that at least one square, $[x,x]$ or $[y,y]$ is missing as a cyclic factor.
	Let $\{x',y'\} = \{x,y\}$ be such that, that $[x',x']$ is not a cyclic factor of $\ws$. We can therefore perform the gluing operation, and obtain a new, strictly shorter word $\ws' \in \lists \{x' \cdot y', y'\}$.  The longer element $x'\cdot y'$ occurs at least twice in $\ws'$, since the number of its occurrences in $\ws'$ is the same as the number of occurrences of $x'$ in $\ws$, the latter word containing both letters at least twice by assumption.
	Moreover, $\ws'$ is primitive, and $\concat \ws' = \concat \ws$ is imprimitive. Therefore, by induction on $\abs \ws$, we have $\ws' \sim [x'\cdot y',x' \cdot y', y']$. In order to show that this is not possible we can successfully reuse the lemma \verb|imprim_ext_suf_comm| mentioned in the proof of Lemma \ref{le:jk_unique}, this time for $u = x'y'x'$ and $v = y'$. The words $u$ and $v$ do not commute because $x'$ and $y'$ do not commute. Since $uv$ is imprimitive,  the word $uvv \sim \concat \ws'$ is primitive. \qed
\end{proof}

This also completes the proof of our main target, Theorem \ref{th:Theorem1}. 
%\section{Proof of Theorem~\ref{th:Theorem1}}
%
%Assume that $\ws \in \lists B$, $\abs {\ws} \geq 2$, $\ws$ is primitive and $\concat \ws$ is imprimitive.
%As $\ws$ is primitive of length at least $2$, both $x$ and $y$ occur in it.
%By~Lemma~\ref{lem:two_squares}, we have that $x$ or $y$ occurs just once in $\ws$.
%Hence, there exist $j$ and $k$ with $j \geq 1$, $k \geq 1$, and $j = 1$ or $k = 1$ such that
%$\ws$ is conjugate with $[x]^j[y]^k$.
%
%Since conjugation preserves primitivity, $[x]^j[y]^k$ is primitive and $\concat [x]^j[y]^k = x^jy^k$ is not.
%We may then use the characterization of the solution of $x^jy^k = z^l$ given by Theorem~\ref{thm:xjykzl_solution} to
%obtain the rest of the claim.
\section{Additional notes on the formalization}

The formalization is a part of an evolving combinatorics on words formalization project.
It relies on its backbone session, called CoW, a version of which is also available in the Archive of Formal Proofs \cite{Combinatorics_Words-AFP}.
This session covers basics concepts in combinatorics on words including the Periodicity lemma.
An overview is available in~\cite{itp2021}.

The evolution of the parent session CoW continued along with the presented results and its latest stable version is available at our repository \cite{CoW_gitlab_v1_6}.
The main results are part of another Isabelle session CoW\_Equations, which, as the name suggests, aims at dealing with word equations.
We have greatly expanded its elementary theory \verb|Equations_Basic.thy| which provides auxiliary lemmas and definitions related to word equations.
Noticeably, it contains the definition \verb|factor_interpretation| (Definition~\ref{def:interpretation}) and related facts.

Two dedicated theories were created: \verb|Binary_Square_Interpretation.thy| and  \verb|Binary_Code_Imprimitive.thy|.
The first contains lemmas and locales dealing with $\{x,y\}$-interpretation of the square $xx$ (for $\abs y \leq \abs x$), culminating in Theorem~\ref{thm:sq_interp}.
The latter contains Theorems~\ref{th:Theorem1}~and~\ref{th:main}.

Another outcome was an expansion of formalized results related to the Lyndon-Sch\"utzenberger theorem.
This result, along with many useful corollaries, was already part of the backbone session CoW, and it was newly supplemented with the parametric solution of the equation $x^jy^k = z^\ell$, specifically Theorem~\ref{thm:xjykzl_solution} and Lemma~\ref{le:jk_unique}.
This formalization is now part of CoW\_Equations in the theory \verb|Lyndon_Schutzenberger.thy|.

Similarly, the formalization of the main results triggered a substantial expansion of existing support for the idea of gluing as mentioned in Section~\ref{sec:gluing}.
Its reworked version is now in a separate theory called \verb|Glued_Codes.thy| (which is part of the session CoW\_Graph\_Lemma).

Let us give a few concrete highlights of the formalization.
A very useful tool, which is part of the CoW session, is the \verb|reversed| attribute.
The attribute produces a symmetrical fact where the symmetry is induced by the mapping \isakw{rev}, i.e., the mapping which reverses the order of elements in a list. 
For instance, the fact stating that if $p$ is a prefix of $v$, then $p$ a prefix of $v \cdot w$, is transformed by the reversed attribute into the fact saying that if $s$ is suffix of $v$, then $s$ is a suffix of $w \cdot v$.
The attribute relies on ad hoc defined rules which induce the symmetry.
In the example, the main reversal rule is 
%{\isachardoublequoteopen}{\isacharparenleft}{\kern0pt}rev\ u\ {\isasymle}p\ rev\ v{\isacharparenright}{\kern0pt}\ {\isacharequal}{\kern0pt}\ {\isacharparenleft}{\kern0pt}v\ {\isasymle}s\ w{\isacharparenright}{\kern0pt}{\isachardoublequoteclose}.
\begin{quote}
 \ {\isachardoublequoteopen}{\isacharparenleft}{\kern0pt}rev\ u\ {\isasymle}p\ rev\ v{\isacharparenright}{\kern0pt}\ {\isacharequal}{\kern0pt}\ u{\isasymle}s\ v{\isachardoublequoteclose}
\end{quote}
\noindent The attribute is used frequently in the present formalization.
For instance, Figure \ref{isa:reversed} shows the formalization of the proof of Cases 1 and 2 of Theorem~\ref{le:jk_unique}. 
Namely, the proof of Case 2 is smoothly deduced from the lemma that deals with Case 1, avoiding writing down the same proof again up to symmetry.

\begin{figure}[t]
\begin{subfigure}{0.49\textwidth}
\begin{tikzpicture}
\node at (0,0)
{
\begin{isaframe}
	\ \ \isacommand{proof}\isamarkupfalse%
	{\isacharparenleft}{\kern0pt}cases{\isacharparenright}{\kern0pt}\isanewline
	\ \ \ \ \isacommand{case}\isamarkupfalse%
	\ {\isadigit{1}}\isanewline
	\ \ \ \ \isacommand{then}\isamarkupfalse%
	\ \isacommand{show}\isamarkupfalse%
	\ {\isacharquery}{\kern0pt}thesis\isanewline
	\ \ \ \ \ \ \isacommand{using}\isamarkupfalse%
	\ LS{\isacharunderscore}{\kern0pt}unique{\isacharunderscore}{\kern0pt}same\isanewline
	\ \ \ \ \ \ \ \ assms{\isacharparenleft}{\kern0pt}{\isadigit{1}}{\isacharcomma}{\kern0pt}\ {\isadigit{4}}{\isacharminus}{\kern0pt}{\isadigit{8}}{\isacharparenright}{\kern0pt}\ \isacommand{by}\isamarkupfalse%
	\ blast \isanewline
	%\end{isaframe}
	%\begin{isaframe}
	\ \ \isacommand{next}\isamarkupfalse%
	\isanewline
	\ \ \ \ \isacommand{case}\isamarkupfalse%
	\ {\isadigit{2}}\isanewline
	\ \ \ \ \isacommand{then}\isamarkupfalse%
	\ \isacommand{show}\isamarkupfalse%
	\ {\isacharquery}{\kern0pt}thesis\ \isanewline
	\ \ \ \ \ \ \isacommand{using}\isamarkupfalse%
	\ LS{\isacharunderscore}{\kern0pt}unique{\isacharunderscore}{\kern0pt}same{\isacharbrackleft}{\kern0pt}reversed{\isacharbrackright}{\kern0pt}\isanewline
	\ \ \ \ \ \ \ \ assms{\isacharparenleft}{\kern0pt}{\isadigit{1}}{\isacharcomma}{\kern0pt}\ {\isadigit{3}}{\isacharcomma}{\kern0pt}\ {\isadigit{5}}{\isacharminus}{\kern0pt}{\isadigit{8}}{\isacharparenright}{\kern0pt}\ \isacommand{by}\isamarkupfalse%
	\ blast
\end{isaframe}
};
\end{tikzpicture}
\caption{Using the \texttt{reversed} attribute to solve symmetric cases.}
\end{subfigure}
\hfill
\begin{subfigure}{0.49\textwidth}
\begin{tikzpicture}
\node at (6,2) 
{
	\begin{isaframe}
		\ \ \isacommand{have}\isamarkupfalse%
		\ {\isachardoublequoteopen}primitive\ {\isacharbrackleft}{\kern0pt}x{\isacharcomma}{\kern0pt}x{\isacharcomma}{\kern0pt}y{\isacharbrackright}{\kern0pt}{\isachardoublequoteclose}\ \isanewline
		\ \ \ \ \isacommand{using}\isamarkupfalse%
		\ {\isacartoucheopen}x\ {\isasymnoteq}\ y{\isacartoucheclose}\isanewline 
		\quad\isacommand{by}\isamarkupfalse%
		\ primitivity{\isacharunderscore}{\kern0pt}inspection
	\end{isaframe}
};
\node at (6,0)
{
	\begin{isaframe}
\ \ \isacommand{from}\isamarkupfalse%
\ {\isacartoucheopen}\isactrlbold {\isacharbar}{\kern0pt}ws\isactrlbold {\isacharbar}{\kern0pt}\ {\isacharequal}{\kern0pt}\ {\isadigit{3}}{\isacartoucheclose}\ {\isacartoucheopen}ws\ {\isasymin}\ lists\ {\isacharbraceleft}{\kern0pt}x{\isacharcomma}{\kern0pt}y{\isacharbraceright}{\kern0pt}{\isacartoucheclose}\isanewline {\isacartoucheopen}x\ {\isasymnoteq}\ y{\isacartoucheclose} 
{\isacartoucheopen}{\isacharbrackleft}{\kern0pt}x{\isacharcomma}{\kern0pt}\ x{\isacharbrackright}{\kern0pt}\ {\isasymle}f\ ws\ {\isasymcdot}\ ws{\isacartoucheclose}\isanewline {\isacartoucheopen}{\isacharbrackleft}{\kern0pt}y{\isacharcomma}{\kern0pt}\ y{\isacharbrackright}{\kern0pt}\ {\isasymle}f\ ws\ {\isasymcdot}\ ws{\isacartoucheclose}\isanewline
\ \ \isacommand{show}\isamarkupfalse%
\ False\isanewline
\ \ \ \isacommand{by}\isamarkupfalse%
\ list{\isacharunderscore}{\kern0pt}inspection\ simp{\isacharunderscore}{\kern0pt}all
	\end{isaframe}
};
\node at (6, -2)
{
	\begin{isaframe}
		\ \ \isacommand{from}\isamarkupfalse%
		\ {\isacartoucheopen}p\ {\isasymcdot}\ t\ {\isasymcdot}\ s\ {\isacharequal}{\kern0pt}\ t\ {\isasymcdot}\ t\ {\isasymcdot}\ p{\isacartoucheclose}\isanewline
		\ \ \isacommand{have}\isamarkupfalse%
		\ {\isachardoublequoteopen}p\ {\isasymcdot}\ t\ {\isacharequal}{\kern0pt}\ t\ {\isasymcdot}\ p{\isachardoublequoteclose}\isanewline
		\ \ \ \ \isacommand{by}\isamarkupfalse%
		\ mismatch
	\end{isaframe}	
};

\end{tikzpicture}
\caption{Methods \texttt{primitivity\_inspection}, \texttt{list\_inspection} and \texttt{mismatch}.}
\end{subfigure}
\caption{Highlights from the formalization in Isabelle/HOL.}\label{isa:reversed}
\end{figure}

To be able to use this attribute fully in the formalization of main results, it needed to be extended to be able to deal with elements of type \texttt{{\isacharprime}a\ list list}, as the constant \verb|factor_interpretation| is of the function type over this exact type.
The new theories of the session CoW\_Equations contain almost 50 uses of this attribute.

The second highlight of the formalization is the use of simple but useful proof methods.
The first method, called \verb|primitivity_inspection|, is able to show primitivity or imprimitivity of a given word.

Another method named \verb|list_inspection| is used to deal with claims that consist of straightforward verification of some property for a set of words given by their length and alphabet.
For instance, this method painlessly concludes the proof of lemma \verb| bin_imprim_both_squares_prim|. The method divides the goal into eight easy subgoals corresponding to eight possible words. All goals are then discharged by \verb|simp_all|.

%Another method named \verb|list_inspection| is used to deal with claims that require to be split into cases based on the length of a relevant list and its alphabet.
%For instance, the straightforward case analysis in the above proof of Lemma~\ref{lem:xy_interp} transforms the main goal into 8 subgoals, enumerating all the possibilities for the list $\mathbf{vs}$ based on its length being 3 and alphabet being $\{x,y\}$. This allows for an easy and clear use of case-specific proofs, which are here required.

The last method we want to mention is \verb|mismatch|.
It is designed to prove that two words commute using the property of a binary code mentioned in Section \ref{sec:notation} and explained in the Appendix \ref{sec:appendix}. Namely, if a product of  words from $\{x,y\}$ starting with $x$ shares a prefix of length at least $\abs{xy}$ with another product of words from $\{x,y\}$, this time starting with $y$, then $x$ and $y$ commute.
Examples of usage of the attribute \texttt{reversed} and all three methods are given in Figure~\ref{isa:reversed}.

\section*{Acknowledgments}

The authors acknowledge support by the Czech Science Foundation grant GA\v CR 20-20621S.

\bibliography{biblio}

\begin{thebibliography}{10}

\bibitem{AlloucheShallit}
Jean-Paul Allouche and Jeffrey Shallit.
\newblock {\em Automatic Sequences: Theory, Applications, Generalizations}.
\newblock Cambridge University Press, USA, 2003.

\bibitem{lerest}
Evelyne Barbin-Le~Rest and Michel Le~Rest.
\newblock Sur la combinatoire des codes \`{a} deux mots.
\newblock {\em Theor. Comput. Sci.}, 41:61--80, 1985.

\bibitem{Berstel1979}
J~Berstel, D~Perrin, J.F Perrot, and A~Restivo.
\newblock Sur le théorème du défaut.
\newblock {\em Journal of Algebra}, 60(1):169 -- 180, 1979.
\newblock URL:
  \url{http://www.sciencedirect.com/science/article/pii/0021869379901133},
  \href {https://doi.org/https://doi.org/10.1016/0021-8693(79)90113-3}
  {\path{doi:https://doi.org/10.1016/0021-8693(79)90113-3}}.

\bibitem{BerstelCodes}
Jean Berstel, Dominique Perrin, and Christophe Reutenauer.
\newblock {\em Codes and Automata}.
\newblock CAMBRIDGE, 2010.
\newblock URL:
  \url{https://www.ebook.de/de/product/8629820/jean_berstel_dominique_perrin_christophe_reutenauer_codes_and_automata.html}.

\bibitem{budkina}
L.~G. Budkina and Al.~A. Markov.
\newblock {$F$}-semigroups with three generators.
\newblock {\em Mat. Zametki}, 14:267--277, 1973.
\newblock Translated from Mat. Zametki 14 (2) (1973) 267–277.

\bibitem{Crochemore1982}
Max Crochemore.
\newblock Sharp characterizations of squarefree morphisms.
\newblock {\em Theoretical Computer Science}, 18(2):221--226, may 1982.
\newblock \href {https://doi.org/10.1016/0304-3975(82)90023-8}
  {\path{doi:10.1016/0304-3975(82)90023-8}}.

\bibitem{FW1965}
N.~J. Fine and H.~S. Wilf.
\newblock Uniqueness theorems for periodic functions.
\newblock {\em Proceedings of the American Mathematical Society},
  16(1):109--109, jan 1965.
\newblock URL: \url{http://dx.doi.org/10.1090/S0002-9939-1965-0174934-9}, \href
  {https://doi.org/10.1090/S0002-9939-1965-0174934-9}
  {\path{doi:10.1090/S0002-9939-1965-0174934-9}}.

\bibitem{terodirk}
Tero Harju and Dirk Nowotka.
\newblock On the independence of equations in three variables.
\newblock {\em Theoret. Comput. Sci.}, 307(1):139--172, 2003.
\newblock Words.
\newblock URL:
  \url{https://doi-org.ezproxy.is.cuni.cz/10.1016/S0304-3975(03)00098-7}, \href
  {https://doi.org/10.1016/S0304-3975(03)00098-7}
  {\path{doi:10.1016/S0304-3975(03)00098-7}}.

\bibitem{itp2021}
\v{S}t\v{e}p\'{a}n Holub and \v{S}t\v{e}p\'{a}n Starosta.
\newblock {Formalization of Basic Combinatorics on Words}.
\newblock In Liron Cohen and Cezary Kaliszyk, editors, {\em 12th International
  Conference on Interactive Theorem Proving (ITP 2021)}, volume 193 of {\em
  Leibniz International Proceedings in Informatics (LIPIcs)}, pages
  22:1--22:17, Dagstuhl, Germany, 2021. Schloss Dagstuhl -- Leibniz-Zentrum
  f{\"u}r Informatik.
\newblock URL: \url{https://drops.dagstuhl.de/opus/volltexte/2021/13917}, \href
  {https://doi.org/10.4230/LIPIcs.ITP.2021.22}
  {\path{doi:10.4230/LIPIcs.ITP.2021.22}}.

\bibitem{lentin}
A.~Lentin and M.-P. Sch\"{u}tzenberger.
\newblock A combinatorial problem in the theory of free monoids.
\newblock In {\em Combinatorial {M}athematics and its {A}pplications ({P}roc.
  {C}onf., {U}niv. {N}orth {C}arolina, {C}hapel {H}ill, {N}.{C}., 1967)}, pages
  128--144. Univ. North Carolina Press, Chapel Hill, N.C., 1969.

\bibitem{lyndon1962}
R.~C. Lyndon and M.-P. Schützenberger.
\newblock The equation $a^m=b^nc^p$ in a free group.
\newblock {\em Michigan Math. J.}, 9(4):289--298, 12 1962.
\newblock \href {https://doi.org/10.1307/mmj/1028998766}
  {\path{doi:10.1307/mmj/1028998766}}.

\bibitem{Manuch}
J{\'{a}}n Manuch.
\newblock Defect effect of bi-infinite words in the two-element case.
\newblock {\em Discret. Math. Theor. Comput. Sci.}, 4(2):273--290, 2001.
\newblock URL: \url{http://dmtcs.episciences.org/279}.

\bibitem{Mitrana1997}
Victor Mitrana.
\newblock Primitive morphisms.
\newblock {\em Inform. Process. Lett.}, 64(6):277--281, December 1997.
\newblock \href {https://doi.org/10.1016/s0020-0190(97)00178-6}
  {\path{doi:10.1016/s0020-0190(97)00178-6}}.

\bibitem{spehner}
J.-P. Spehner.
\newblock {\em {Quelques probl\`{e}mes d'extension, de conjugaison et de
  presentation des sous-mono\"{i}des d'un mono\"{i}de libre.}}
\newblock PhD thesis, Universit\'{e} Paris VII, Paris, 1976.

\bibitem{Combinatorics_Words-AFP}
Štěpán Holub, Martin Raška, and Štěpán Starosta.
\newblock Combinatorics on words basics.
\newblock {\em Archive of Formal Proofs}, May 2021.
\newblock \url{https://isa-afp.org/entries/Combinatorics_Words.html}, Formal
  proof development.

\bibitem{CoW_gitlab_v1_6}
Štěpán Holub, Martin Raška, and Štěpán Starosta.
\newblock Combinatorics on words formalized (release v1.6).
\newblock \url{https://gitlab.com/formalcow/combinatorics-on-words-formalized},
  2022.

\end{thebibliography}

\appendix

\section{Background results in combinatorics on words} \label{sec:appendix}

One of the main advantages of formalization is that it allows to tune the level of detail without compromising precision and correctness.
In the present paper we tried to keep the exposition on the level which should make all claims convincing for a researcher in combinatorics on words. This Appendix can help readers who seek more information about the folklore and other well known results.

Most ideas in combinatorics on words are related to periodicity. A word $w$ has a periodic root $r$ if it is a prefix of repeated occurrences of $r$. This can be expressed as $w \leq_p r^\omega$, or equivalently, and using finite words only, as $w \leq_p r \cdot w$, provided $r$ is nonempty. Note that the periodic root $r$ of $w$ need not be primitive, but it is always possible to consider the corresponding primitive root $\rho\, r$, which is also a periodic root of $w$. Note that any word has infinitely many periodic roots since we allow $r$ to be longer than $w$. Nevertheless, a word can have more than one period even if we consider only periods shorter than $|w|$. Such a possibility is controlled by the Periodicity lemma, often called the Theorem of Fine and Wilf (see \cite{FW1965}):

\begin{lemma}[\texttt{per\_lemma\_comm}] \label{perlem}
	If $w$ has a period $u$ and $v$, i.e., $w \leq_p uw$ and $w \leq_p vw$, with $|u|+|v| - \gcd(|u|,|v|) \leq |w|$, then $uv = vu$.
\end{lemma}
Usually, the weaker test $|u| + |v| \leq  |w| $ is sufficient to indicate that $u$ and $v$ commute.

Conjugation $u \sim v$ of two words is defined by the existence of two words $r$ and $q$ such that $u = rq$ and $v = qr$. Two words are conjugate if they are the same up to rotation, therefore the conjugacy class is naturally seen as a word understood cyclically. A word $z$ is then called a ``cyclic factor'' of $u$ if it is a factor of some conjugate of $u$, which is equivalent to being a factor of the square of $u\cdot u$ if $|z| \leq |u|$.
Importantly, conjugation $u \sim v$ is also characterized as follows:
\begin{lemma}[\texttt{conjugation}]\label{conjugation}
	If $uz = zv$ for nonempty $u$,  then there exists words $r$ and $q$ and an integer $k$ such that
	\[
	u = rq, \quad v = qr \quad \text{ and } \quad z = (rq)^kr.
	\]
\end{lemma}

We have said that $w$ has a periodic root $r$ if it is a prefix of $r^\omega$. If $w$ is a factor, not necessarily a prefix, of $r^\omega$, then it has a periodic root which is a conjugate of $r$. In particular, if $\abs u = \abs v$, then $u \sim v$ is equivalent to $u$ and $v$ being mutually factors of a power of the other word.

Commutation of two words, that is, equality $u \cdot v = v \cdot u$, is characterized as follows: 
\begin{lemma}[\texttt{comm}]\label{commutation}
	If $xy = yx$ if and only if $x = t^k$ and $y = t^m$ for some word $t$ and some integers $k,m \geq 0$.
\end{lemma}
Since every nonempty word has a (unique) primitive root, the word $t$ above can be chosen primitive ($k$ or $m$ can be chosen $0$ if $x$ or $y$ is empty). 

We mention that the given characterizations of conjugation and commutation (in a slightly expanded form) are called ``the first theorem of Lyndon and Sch\"utzenberger'' and 
``the second theorem of Lyndon and Sch\"utzenberger'' respectively in \cite{AlloucheShallit}, while we reserve the term ``the theorem of Lyndon and Sch\"u\-tzen\-ber\-ger'' to the following fact often used in the paper:
\begin{theorem}[\texttt{Lyndon\_Schutzenberger}] \label{thm:LS}
	If $x^jy^k = z^\ell$ with $j \geq 2$, $k \geq 2$ and $\ell \geq 2$, then the words $x$, $y$ and $z$ commute.
\end{theorem}

A crucial property of a primitive word $t$ is that it cannot be a nontrivial factor of its own square. For a general word $u$, the equality $u\cdot u = p \cdot u \cdot s$ with nonempty $p$ and $s$ implies that all three words $p$, $s$, $u$ commute, that is, have a common primitive root $t$. This can be seen by writing $u = t^k$, and noticing that the presence of a nontrivial factor $u$ inside $uu$ can be obtained exclusively by a shift by several $t$'s. This is the ``synchronization'' idea mentioned in the paper.

One of the typical applications of this idea is the fact that if $w$ has a periodic root $r$, and at the same time $r$ is a suffix of $w$, then $w$ and $r$ commute. Another application is the proof of the fact that $qrqr$ cannot be a factor of a power of  $qrrq$ unless $q$ and $r$ commute, used in the proof of Lemma \ref{le:jk_unique}. This can be seen as follows. If $qrqr$ is a factor of $(qrrq)^\omega$, then $qrrq \sim qrqr$ since the words are of the same length. This means that also $qrrq$ is a factor of $(qr)^\omega$. Synchronizing the prefix $qr$ of $qrrq$ within $(qr)^\omega$ we obtain that the rest of the word, the word $rq$, is a prefix of $qr$, and the conclusion follows.
\medskip

Let $x$ and $y$ be two words that do not commute.
The longest common prefix of $xy$ and $yx$ is denoted $\alpha$. Let $c_x$ and $c_y$ be the letter following $\alpha$ in $xy$ and $yx$ respectively. A crucial property of $\alpha$ is that it is a prefix of any sufficiently long word in $\gen{\{x,y\}}$. Moreover, if $\ws = [u_1,u_2,\ldots,u_n] \in \lists \{x,y\}$ is such that $\concat \ws$ is longer than $\alpha$, then $\alpha\cdot [c_x]$ is a prefix of $\concat \ws$ if $u_1 = x$ and  
 $\alpha\cdot [c_y]$ is a prefix of $\concat \ws$ if $u_1 = y$. That is why the length of $\alpha$ is sometimes called ``the decoding delay'' of the binary code $\{x,y\}$. Note that the property indeed in particular implies that $\{x,y\}$ is a code, that is, it does not satisfy any nontrivial relation. It is also behind our method \verb|mismatch|. Finally, using this property, the proof of Lemma \ref{lem:lcp} is straightforward.  

\end{document}